\documentclass[journal]{IEEEtran}
\usepackage{mathrsfs}
\usepackage{url,times,amsmath,color,amssymb,graphicx,epsfig,psfig,cite,psfrag,subfigure,dsfont,booktabs}
\usepackage{algorithm}
\usepackage{algorithmic}
\usepackage{bm}
\usepackage{color}
\usepackage{multicol}
\usepackage{stfloats}

\newtheorem{mypro}{Proposition}
\begin{document}
\title{Energy Efficiency Optimization for NOMA UAV Network with Imperfect CSI}

\author{Haijun Zhang,~\IEEEmembership{Senior Member,~IEEE}, Jianmin Zhang, Keping Long,~\IEEEmembership{Senior Member,~IEEE}

\thanks{This work is supported by the National Natural Science Foundation of China (61822104, 61771044), Beijing Natural Science Foundation (L172025, L172049), 111 Project (No. B170003), and the Fundamental Research Funds for the Central Universities(FRF-TP-19-002C1, RC1631), Beijing Top Discipline for Artificial Intelligent Science and Engineering, University of Science and Technology Beijing.  This paper was presented in part at the IEEE International Conference on Communications (ICC 2020), Dublin, Ireland, 2020. The corresponding authors are Keping Long and Haijun Zhang.

Haijun Zhang, Jianmin Zhang, and Keping Long are with Institute of Artificial Intelligence, Beijing Advanced Innovation Center for Materials Genome Engineering, Beijing Engineering and Technology Research Center for Convergence Networks and Ubiquitous Services, University of Science and Technology Beijing, Beijing 100083, China (e-mail: haijunzhang@ieee.org, zhangjianmin96@163.com, longkeping@ustb.edu.cn).

}} \maketitle
%##################################################################

\thispagestyle{empty} % no page number for the first page

\begin{abstract}
  Unmanned aerial vehicles (UAVs) are developing rapidly owing to flexible deployment and access
  services as air base stations. However, the channel errors of low-altitude communication links formed by
  mobile deployment of UAVs cannot be ignored. And the energy efficiency of the UAVs communication
  with imperfect channel state information (CSI) hasn¡¯t been well studied yet. Therefore, we focus on
  system performance optimization in non-orthogonal multiple access (NOMA) UAV network considering
  imperfect CSI between the UAV and users. A suboptimal resource allocation scheme including user
  scheduling and power allocation is designed for maximizing energy efficiency. Because of the nonconvexity
  of optimization function with an probability constraint for imperfect CSI, the original problem
  is converted into a non-probability problem and then decoupled into two convex subproblems. First, a
  user scheduling method is applied in the two-side matching of users and subchannels by the difference of
  convex programming. Then based on user scheduling, the energy efficiency in UAV cells is optimized
  through a suboptimal power allocation algorithm by successive convex approximation method. The
  simulation results prove that the proposed algorithm is effective compared with existing resource
  allocation schemes.
\end{abstract}
%\vspace{-1mm}

\section{Introduction}
Because of the sharp growth of communication devices, researches on unmanned aerial vehicle
(UAV) networks have attracted attention for their flexible deployment. UAV can meet the
requirements for global seamless coverage of the fifth generation (5G) or even beyond fifth
generation (B5G) mobile networks \cite{1UAV}. In order to meet the wireless traffic needs of all users
as much as possible, especially in urban areas with large-scale mobile users, 5G requires high-capacity
dense access point coverage \cite{2Network}. In areas where users and equipments are heavily deployed,
UAVs can act as air base stations (BS) to assist ground BS to provide services. The application
of UAVs not only increases the coverage area, but also enhances the system connectivity. In
addition, many Internet of Things (IoT) devices are very vulnerable to disasters. If the ground
BS is damaged and cannot provide communication services, UAVs can be used for emergency
communication \cite{3DSF}.

Different with the channel model in small cell networks, the existence
of direct line-of-sight (LOS) need to be considered in the low-altitude communication link.
Extra reflections caused by buildings in non-line-of-sight (NLOS) also need to be considered
\cite{4Propagation}. For the study of any communication environment, an accurate channel model is necessary \cite{104Trajectory}. Therefore, there are many studies on the channel model of the air-to-ground communication. \cite{5Optimal} studied the channel model from low altitude platforms (LAPs) to users. It also proved that the LOS probability from UAVs to users is in connection with the elevation angle and parameters
of the located area including the distribution of obstacles. In \cite{6Modeling}, the authors focused on the
impact of elevation angle because the elevation angle has the opposite effect on LOS probability
and antenna gain. And it is obvious that the channel gain from the UAV to users is related to
the UAV height.

UAVs can support enhanced mobile broadband (eMBB) \cite{106High} to assist with 5G cellular network
connections \cite{7Distributed}. The traffic load of the macrocell can be transferred to the small cell access
point to alleviate the burden on the cellular network \cite{8Optimal}, and the UAV can play the same role
as the air small cell. UAVs that can move and provide links at any time are an important part
of 5G interconnections at low altitudes. However, considering the limitations of UAV launch
power and flight height, the deployment of numerous UAVs in the B5G communication system
makes the system performance of UAV communication networks worth studying. \cite{9Throughput} considered
UAVs as mobile access points, and maximized user throughput by optimizing trajectory and
resource allocation. Multiple UAV collaborations are widely used in the IoT and wireless sensor networks\cite{112An}\cite{113Flight}. UAVs with underlaid Device-to-Device (D2D) communications is considered
in \cite{10Unmanned}, the author discussed the coverage probability and overall rates in the entire UAV system.
Further considering D2D communication with energy harvesting, resource optimization in the
UAV auxiliary network to maximize throughput per unit time was discussed in \cite{11Resource}.

Non-orthogonal multiple access (NOMA) which is a very critical technology of 5G \cite{13Energy},
has been applied to UAV communication to provide services for more ground users. NOMA
communication between the aerial UAV and ground is a promising technology called groundaerial
NOMA \cite{14Cellular}. NOMA technology ensures that multiple users transmit simultaneously in
a subchannel by successive interference cancellation (SIC) technology. Specifically, users with strong channel gain can eliminate the signal from users with poor channel gain \cite{114Closed}. Using SIC on user side
can successfully decode the received signal according to the channel gain order and improve
system performance \cite{15Secure}.

Many works of resource allocation have been done in the UAV network. In \cite{16Placement}, since
the total system rate in UAV network was non-convex, the author first optimized the position of
the UAV to minimize the path loss and then optimized the power allocation. \cite{17Non} discussed the
power allocation in the case of UAV height fixation and UAV altitude movement, and proved
that NOMA-assisted UAV network had a better performance than orthogonal frequency division
multiple access (OFDMA)-assisted UAV network. \cite{18Cellular}\cite{19Multi} focused on the interference coordination of the UAV-based wireless networks.

\cite{20A} comprehensively optimized the resource allocation of macro BS and UAVs, and then used the difference of convex (DC) program to solve the non-convex objective functions. In previous researches, if
the stability of the network was considered, probability of outage would be introduced. If the
instantaneous data rate exceeds the maximum capacity of the system, it is considered as a
communication outage. The author in \cite{21Multiple} confirmed that the outage probability was related to
the users quality of services (QoS) demand and power allocation.
%The author in \cite{22On} studied the outage probability of users in NOMA UAV with D2D communications.

However, previous studies of the UAV communication network only discussed the perfect
channel state information (CSI) with path loss. The author in \cite{23Energy} fixed the height of the aircraft
and maximized system energy efficiency through optimization of the flight trajectory of the
aircraft. \cite{24An} considered the UAV-assisted 5G network with macro BS and connected users with
different UAVs. It proved that the introduction of UAV can enhance the system performance, and
the energy efficiency was closely related to the flight height of the UAV. However, different from
\cite{5Optimal}, the small-scale fading of UAV channels should also be considered. Because in this paper,
the low-altitude UAV communication is considered in urban areas where users and buildings
are densely populated. In this scenario, there must be small-scale fading caused by multipath
propagation, and it cannot be ignored in comparison to path loss.
Because UAVs communicate with ground users through air links, the CSI is not perfectly
perceived in practice, due to estimation errors and finite data feedback \cite{124Distributed}. \cite{25Joint} discussed the resource allocation problem of the cellular
network in NOMA with imperfect CSI. Due to the high-speed mobility of UAVs, UAVs can be
randomly deployed over ultra-dense users with demand. The low-altitude communication link
is affected by the reflection of various obstacles. Therefore, compared with the ground BS to
the user, it is more necessary for UAV communication to consider about the imperfection of the
channel. Then, energy efficiency of users in the NOMA UAV network with imperfect CSI is
discussed in this paper.

In this work, system performance in a downlink NOMA network architecture with UAVs and
a ground macro BS is studied. To the best of our knowledge, the energy efficiency optimization
with imperfect CSI in NOMA UAV network hasn¡¯t been studied yet. Different with \cite{26Two} that
studied two-side matching between two types of users sharing spectrum to improve spectrum
efficiency, a two-side matching between users and subchannels is studied to improve energy
efficiency. Subchannel and power resources are successively allocated to optimize the UAV
network energy efficiency. Considering imperfect CSI, an estimation error variance is introduced
to form energy efficiency expression \cite{27Effect}. UAV moves to users who cannot be served by a
ground BS, and a resource allocation method is designed based on the number of users and
communication distance when UAV is hovering. As for resource optimization method, first the
power of each subchannel is set to equal when users and subchannels are matched. Then the
appropriate power is allocated on different subchannels. The main contributions of our paper are
listed as follows.

\begin{itemize}
  \item \emph{Energy efficiency optimization for NOMA UAV network:} The total network consists of UAVs
  acting as air BSs, a ground BS, and users receiving signals from UAVs and the ground BS
  respectively. Unlike the optimization of total data rate, our target is to optimize energy
  efficiency of UAV users with NOMA. The energy efficiency is optimized through user
  scheduling as well as power allocation in the UAV network, considering the constraints
  such as imperfect CSI, transmit power limit of the UAV, and the interference of the UAV
  to the macro user. The non-convex target problem is converted to two subproblems with
  convexity, then solved by user scheduling and power allocation respectively.
  \item \emph{Two-side matching between subchannels and users by user scheduling algorithm:} The twoside
  suboptimal selection algorithm of the subchannel and the user is studied on the basis of
  the channel gain with imperfect CSI from UAVs to users. The mobility of the UAV results
  in different distances of the communication link, which will cause changes in channel and
  user matching. Therefore, two-side matching is performed with hovering UAV. In the paper, assume that at most two users can be assigned to each subchannel. Then a DC
  programming is applied to select the energy efficient user pair on each subchannel by power
  proportion factor.
  \item \emph{Design of suboptimal power allocation algorithm:} On the basis of user scheduling, matching
  of users and subchannels has been completed, and the power proportion factor on each
  subchannel has been fixed, only the power on each subchannel has not been allocated,
  but the energy efficiency in UAV network is non-convex with respect to power. Then the
  successive convex approximation method is applied to transform the optimization problem
  to convex problem. Finally, a suboptimal power allocation algorithm is designed to optimize
  energy efficiency.
\end{itemize}

The rest of our paper is organized as follows. Section II depicts the system model of UAV
network. In Section III, we analyze the representation of the energy efficiency function and introduce
user scheduling and power allocation methods. Then the algorithm of resource allocation is shown in Section IV. Simulation results of energy efficiency are presented in Section V. At last, Section VI is the conclusion of this paper.

\section{System Model}
\begin{figure}
        \centering
        \includegraphics*[width=8cm]{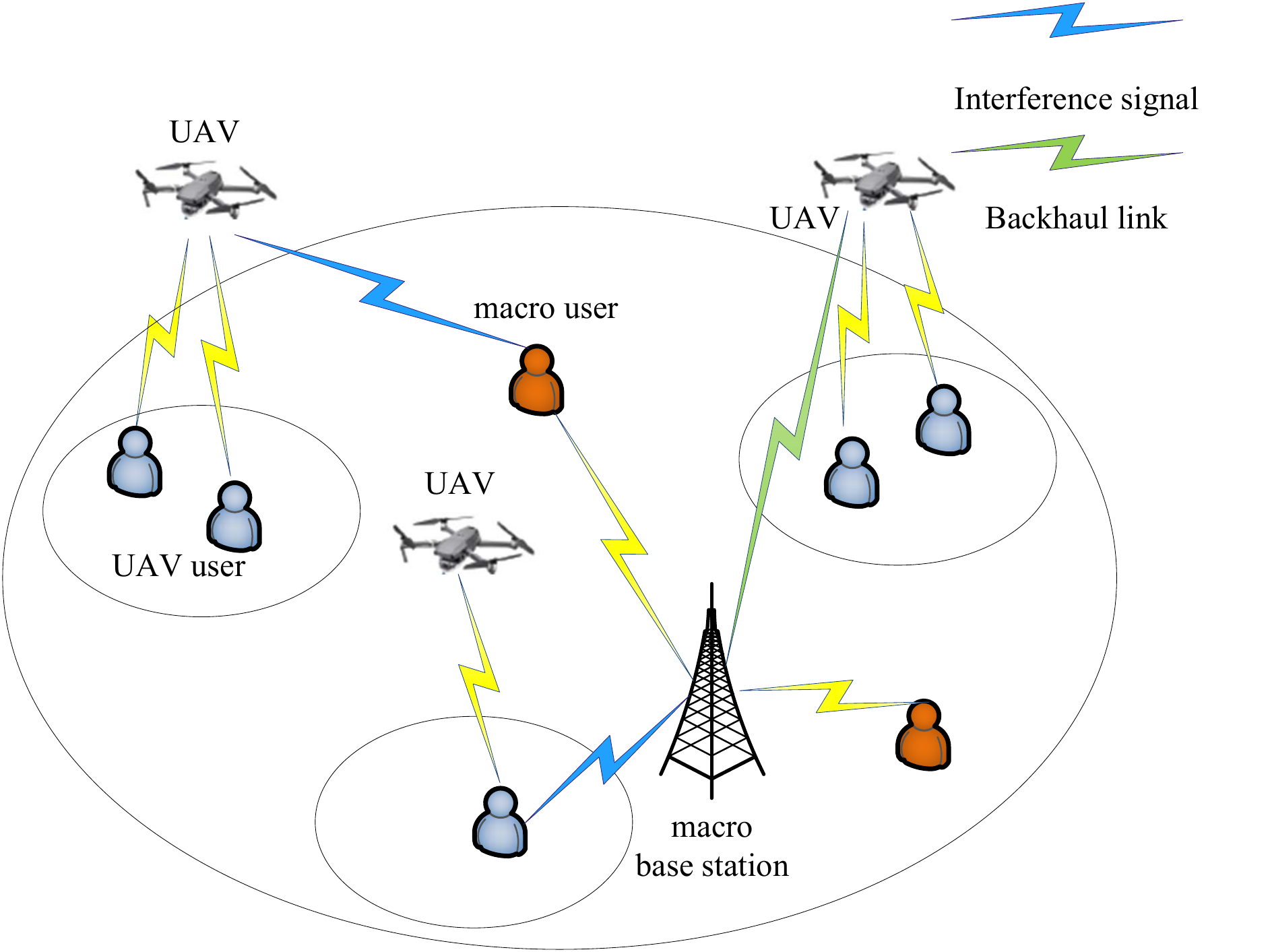}
        \caption{Topology of the UAV network with a macro base station.}
        \label{fig:1}

\end{figure}

The network where one macro BS and I UAVs coexist is shown in Fig. 1. A hovering UAV only serves a UAV cell and $I$ UAV small cells are located in the macrocell range with radius R, $W$ macro users are randomly located in macrocell. $N$ users are provided communication services by UAVs which are randomly located in the UAV cell called UAV users.

The same spectrum are shared between all UAV cells and the macrocell. Let ${h_i}$ represents the altitude of the $i$th UAV. The bandwidth of the channel $BW$ can be divided into $K$ subchannels. Therefore, each subchannel bandwidth ${B_{sc}} = BW/K$. The index of the subchannel is $k$ where $k \in \left\{ {1,2,...,K\} } \right.$.

Let ${N_k} \in \left\{ {{1}} \right.,{2},...{N_k}\} $ express as the users number assigned on the subchannel $k$ ($S{C_k}$) where $N = {N_1} + {N_2} + ... + {N_K}$. The $n$th user assigned on a subchannel is denoted as $n$. Assume that at most two users can be assigned to each subchannel to reduce computational complexity. That is, ${N_k} <= 2$. The power of the $n$th user on $S{C_k}$ in $i$th UAV network is denoted by ${p_{n,i,k}}$. The transmit power of a UAV is ${P_{UAV}}$, then the subchannel
and UAV power limit are shown as $\sum\limits_{n = 1}^{{N_k}} {{p_{n,i,k}} = {p_{n,i}}} $ and $\sum\limits_{k = 1}^K {{p_{n,i}}}  \le {P_{UAV}}$, where ${{p_{n,i}}}$ is the
power allocated on $S{C_k}$. Also, $p_{w,k}^M$ is the BS power allocated to macro user $w$ on subchannel $k$ with limited $\sum\limits_{k = 1}^K {\sum\limits_{w = 1}^W {{p_{w,k}^M}} }  \le {P_{BS}}$, where ${P_{BS}}$ is the maximum BS power.

The signal transmitted by the $i$th UAV network through $S{C_k}$ is expressed as
  \begin{equation}
  {x_{i,k}} = \sum\limits_{n = 1}^{{N_{k}}} {\sqrt {{p_{n,i,k}}} {s_{n,i}}},
  \end{equation}
where ${{s_{n,i}}}$ is the modulated symbol.

The $n$th user of the $i$th UAV receives the signal on subchannel $k$ as
  \begin{equation}
  \begin{split}
  &{y_{n,i,k}} = {H_{n,i,k}}{x_{i,k}} + z_{n,i,k} \\
  &= \sqrt {{p_{n,i,k}}} {H_{n,i,k}}{s_{n,i}} + \sum\limits_{j = 1,j \ne n}^{{N_k}} {\sqrt {{p_{j,i,k}}} {H_{n,i,k}}{s_{j,i}}}  + z_{n,i,k},
  \end{split}
  \end{equation}
where ${H_{n,i,k}} = PL{(d)_{n,i,k}}{g_{n,i,k}}$, ${g_{n,i,k}} \sim CN(0,1)$ is a complex Gaussian random variable representing the small-scale fading in the subchannel $k$ of the $i$th UAV to the $n$th user. $PL{(d)_{n,i,k}}$ is the path loss from $i$th UAV to $n$th user with distance $d$. The term ${z_{n,i,k}} \sim CN(0,\sigma ^2)$ is the white Gaussian noise which has a mean zero and variance $\sigma ^2$.

Assume that the same subchannel has the same channel fading which may be different across subchannels. Unlike previous articles that only considered UAV-to-user path loss, small-scale
fading is also considered in our paper. Because in a user-intensive area, if a disaster occurs,
the low-altitude UAV can be served as an air BS and obstacles block the LOS transmission.
Small-scale fading occurs due to the obstacles causing multipath transmission in the subchannel.
However, differ with the channel fading from the BS to macro users, the path loss between the
hovering UAV and UAV users including LOS and NLOS condition. The part among UAVs
and urban buildings in free space is LOS, and the remaining part from the obstacle to users is
NLOS. Because of the shadow effect and obstructions of dense urban buildings, fading in NLOS
part is much larger than the LOS condition.

As for the probability of LOS in the low-altitude link between hovering UAVs and users depends upon the surroundings, the elevation angle, and the distribution of users and UAVs. The probability of LOS link is written as \cite{5Optimal}
  \begin{equation}
  {P_{LOS}} = {1 \over {1 + u\exp ( - v[\phi  - A])}},
  \end{equation}
where $u$ and $v$ are environment parameters including the height of obstacles, etc. The elevation angle $\phi $ is expressed as
  \begin{equation}
  \phi   = {{180} \over \pi } \times {\arcsin }({h \over d}),
  \end{equation}
where $h$ represents the hovering height of UAVs, $d$ represents the distance from UAV to users.
Assume that the path loss between UAVs and the user strictly follows both LOS and NLOS
\cite{5Optimal}. Then, the NLOS link probability is presented as
  \begin{equation}
  {P_{NLOS}} = 1 - {P_{LOS}}.
  \end{equation}

Therefore, the total path loss from UAVs to users is shown as follows\cite{10Unmanned}
  \begin{equation}
  PL(d) = {P_{LOS}} \times {(d)^{ - \alpha }} + {P_{NLOS}} \times \eta {(d)^{ - \alpha }},
  \end{equation}
where $\alpha $ represents path loss index between users and UAV connection, and $\eta $ is an extra attenuation coefficient because of the NLOS link.

If users can perfectly perceive the CSI in the UAV low-altitude communication channel, the
signal-to-interference-plus-noise ratio (SINR) in the $i$th UAV network for UAV user $n$ occupying
the kth subchannel is
  \begin{equation}
  SIN{R_{n,i,k}} = {{{p_{n,i,k}}|{H_{n,i,k}}{|^2}} \over {p_k^M|H_{n,i,k}^M{|^2} + \sum\limits_{j = 1,j \ne n}^{{N_k}} {{p_{j,i,k}}|{H_{n,i,k}}{|^2}}  + {\sigma}^2}},
  \end{equation}
${H_{n,i,k}^M}$ is channel gain between macro BS and $i$th UAV user $n$ through subchannel $k$. ${p_k^M|H_{n,i,k}^M{|^2}}$ represents the interference suffered by UAV users on $S{C_k}$ from BS. ${\sum\limits_{j = 1,j \ne n}^{{N_k}} {{p_{j,i,k}}|{H_{n,i,k}}{|^2}} }$ represents the interference resulted by the other users in the same subchannel $k$ due to NOMA. It is verified in \cite{28Geometrical} that the directional antenna has high channel correlation in the mobile-to-mobile channel. Therefore, the users of the UAV cell are considered to use directional antennas, the UAV-to-UAV interference can be negligible compared to cross-tier macro-to-UAV interference \cite{29Joint}.

The SIC technology is widely used in NOMA networks to reduce interference between co-frequency users. NOMA system assigns much power to users having lower channel gain. So the user with a better channel gain condition can eliminate interference created by other users in a poorer channel condition on a subchannel. Consider that for $i$th UAV, the gain of ${N_k}$ users on $k$th subchannel is ordered as
  \begin{equation}
  |{H_{1,i,k}}| \le |{H_{2,i,k}}| \le ...{\rm{ }}|{H_{n,i,k}}| \le ... \le |{H_{{N_k},i,k}}|.
  \end{equation}

So that for user $j > n$, then $|{H_{n,i,k}}| \le |{H_{j,i,k}}|$, $n$th UAV user is able to decode the signal
because interference created by user $j$ can be removed. If user $j < n$, signal from user $j$ is
considered interference. Therefore, when the receiver applies SIC technology, the SINR of user
$n$ is rewritten as

  \begin{equation}
  SIN{R_{n,i,k}} = {{{p_{n,i,k}}|{H_{n,i,k}}{|^2}} \over {p_k^M|H_{n,i,k}^M{|^2} + \sum\limits_{j = 1}^{n - 1} {{p_{j,i,k}}|{H_{n,i,k}}{|^2}}  + {\sigma }^2}}.
  \end{equation}

According to Shannon Theory, the maximum capacity of the $n$th user on $S{C_k}$ in $i$th UAV network is
  \begin{equation}
  {C_{n,i,k}} = {B_{sc}}{\log _2}(1 + SIN{R_{n,i,k}}).
  \end{equation}

 Because of large scale fading such as path loss changes slowly, we assume that the BS as well as UAVs can estimate it perfectly \cite{30Impact}. Therefore, in this paper, imperfect CSI especially as small scale fading between hovering UAVs and users is considered. The energy efficiency optimization in UAV network is studied with the estimated fading channel. The small scale fading coefficient for the $i$th UAV to the $n$th user on $k$th subchannel is shown as
  \begin{equation}
  {g_{n,i,k}} =  {\hat{g}_{n,i,k}} + {e_{n,i,k}},
  \end{equation}
where channel estimated error ${e_{n,i,k}}$ is a complex Gaussian distribution that has zero-mean and variance $\sigma _e^2$. And ${\hat{g}_{n,i,k}} \sim CN(0,1 - \sigma _e^2)$ represents the estimated small-scale fading channel coefficient. The estimated channel gain is
 \begin{equation}
 {{\hat H }_{n,i,k}} = PL{(d)_{n,i,k}}{\hat{g}}_{n,i,k}= PL{(d)_{n,i,k}}{\hat{g}_{n,i,k}} + {e_{n,i,k}}.
 \end{equation}
 Then channel estimated error and known CSI are combined to reorder imperfect channel gain as follows
   \begin{equation}
   \label{order}
  |{{\hat H }_{1,i,k}}| \le |{{\hat H }_{2,i,k}}| \le ...{\rm{ }}|{{\hat H }_{n,i,k}}| \le ... \le |{{\hat H }_{{N_k},i,k}}|.
  \end{equation}
According to SIC, the estimated ${{\hat{SINR}} _{n,i,k}}$ and data rate are written as
  \begin{equation}
  {{\hat{SINR}} _{n,i,k}} = {{{p_{n,i,k}}|{{\hat H }_{n,i,k}}{|^2}} \over {p_k^M|H_{n,i,k}^M{|^2} + \sum\limits_{j = 1}^{n - 1} {{p_{j,i,k}}|{{\hat H}_{n,i,k}}{|^2}}  + {\sigma ^2}}},
  \end{equation}
  \begin{equation}
  {R_{n,i,k}} = {B_{sc}}{\log _2}(1 + {{\hat{SINR}} _{n,i,k}}).
  \end{equation}

Considering the channel changes caused by the UAV position and the estimation error
in subchannels, the outage probability is introduced to measure system performance. Therefore, the average
outage total rate is depicted as \cite{31Energy}
  \begin{equation}
  \hat {R}  = \sum\limits_{i = 1}^I {\sum\limits_{k = 1}^K {\sum\limits_{n = 1}^{{N_k}} {{R_{n,i,k}}(Pr [{C_{n,i,k}} > {R_{n,i,k}}|{{\hat g}_{n,i,k}}]} } } ).
  \end{equation}

In the case of ${\hat g}_{n,i,k}$, there is a non-zero outage probability. Therefore, the total rate of successful transmission is the average outage probability.

Then, the energy efficiency of user $n$ on subchannel $k$ in UAV $i$ is
  \begin{equation}
  {E{E_{n,i,k}}}={{{(Pr [{C_{n,i,k}} > {R_{n,i,k}}|{{\hat g }_{n,i,k}}]){R_{n,i,k}}} \over {{p_m} + {p_{n,i,k}}}}}.
  \end{equation}
where ${{p_m}}$ represents the mechanical energy consumption per unit time to keep hovering against gravity, and is generally considered as constant \cite{032Areial}.

Besides the mechanical energy consumption by the hovering of the UAV, the transmit power for providing communication services to ground users is also part of the energy consumption. However, the power of each UAV is limited, it is essential to achieve greater energy efficiency by saving transmit power. The energy efficiency objective function of our paper is to optimize the limited transmit power to achieve a larger system rate.
Considering all users and all subchannels occupied in UAV cells, the total energy efficiency in
our paper should be depicted as
  \begin{equation}
  \begin{split}
  &EE(U,P) = \sum\limits_{i = 1}^I {\sum\limits_{k = 1}^K {\sum\limits_{n = 1}^{{N_k}} {E{E_{n,i,k}}} } } \\
  &= \sum\limits_{i = 1}^I {\sum\limits_{k = 1}^K {\sum\limits_{n = 1}^{{N_k}} {{{(Pr [{C_{n,i,k}} > {R_{n,i,k}}|{{\hat g }_{n,i,k}}]){R_{n,i,k}}} \over {{p_m} + {p_{n,i,k}}}}} } },
  \end{split}
  \end{equation}
 $U$ represents the matrix of users and
subchannels matching, and $P$ represents the matrix of power allocation. The energy efficiency
problem in the UAV network are resolved by optimizing $U$ and $P$.

The energy efficiency optimization in the UAV cells with imperfect CSI should consider restrictions as follows:
\begin{itemize}
  \item \emph{Transmit power constraint of each UAV:}
    \begin{equation}
    {\rm{ }}\sum\limits_{k = 1}^K {\sum\limits_{n = 1}^{{N_k}} {{p_{n,i,k}}} }  \le {P_{UAV}},{\rm{  }}\forall i
    \end{equation}
    where $P_{UAV}$ is the maximum transmit power of each UAV.
  \item \emph{Interference limitation for macro users:}
    \begin{equation}
    {\rm{ }}\sum\limits_{i = 1}^I {{p_{i,k}}|H_{w,i,k}^M{|^2} \le {I_k}} ,{\rm{ }}\forall w,k{\rm{  }}
    \end{equation}
    On the one hand, multiple UAVs provide access services for more ground users. On the other hand, UAVs bring inevitable interference to users who are still using ground BS for communication service. Therefore, the interference restrictions of macro users also need to be considered.
  \item \emph{Outage probability constraint:}
    The UAV cannot provide excellent service to the user when an outage occurs. In order to guarantee the QoS, it is important to consider the outage probability constraint.
    \begin{equation}
    {\rm{ }}\Pr [{C_{n,i,k}} < {R_{n,i,k}}|{\hat{g} _{n,i,k}}] \le {\varepsilon _{out}},{\rm{ }}\forall i,n,k.
    \end{equation}
\end{itemize}

Therefore, the optimization problem with related constraints in the downlink UAV cells is formulated as

\begin{equation}
\label{originalobj}
\max {\rm{ }}\mathop {EE}\limits_{(U,P)} (U,P)
\end{equation}
\begin{equation}
\label{originalconstraints}
\begin{aligned}
& \text{s.t.}
& & {C1:{\rm{ }}\sum\limits_{k = 1}^K {\sum\limits_{n = 1}^{{N_k}} {{p_{n,i,k}}} }  \le {P_{UAV}},{\rm{  }}\forall i }\\
&&& {C2:{\rm{ }}{p_{n,i,k}} \ge 0,{\rm{ }}\forall i,n,k }\\
&&& {C3:{\rm{ }}\sum\limits_{i = 1}^I {{p_{i,k}}|H_{w,i,k}^M{|^2} \le {I_k}} ,{\rm{ }}\forall w,k{\rm{  }} }\\
&&& {C4:{\rm{ }}\Pr [{C_{n,i,k}} < {R_{n,i,k}}|{\hat{g} _{n,i,k}}] \le {\varepsilon _{out}},{\rm{ }}\forall i,n,k}\\
&&& {C5:{\rm{ }}{N_k} \le 2.}\\
\end{aligned}
\end{equation}
$C1$ is total transmission power constraint of each UAV cell; $C2$ represents non-negative power
assigned to each UAV user, and ensures that each user can perform normal communication
services; $C3$ is the interference restriction imposed by the UAV network on the macro user with
maximum tolerable interference level ${{I_k}}$ per macro user; $C4$ expresses the outage probability
threshold ${\varepsilon _{out}}$ and $C5$ means that at most two users are assigned on each subchannel considering multiplexing due to computational complexity in our paper.

\section{Energy Efficiency Optimization and Resource Allocation}

Since the objective function in (\ref{originalobj}) and $C4$ in (\ref{originalconstraints}) are non-convex, we introduce an outage threshold to remove probability constraints from problems, and rewrite the energy efficiency
function in the UAV network. The target problem is decoupled into subchannel-user match and
power control. In other words, power is allocated based on user scheduling to optimize energy
efficiency considering imperfect CSI.

\subsection{Optimization Problem Transformation}

The original problem has a probability constraint and belongs to a non-convex function. To remove the probability constraint, the target problem is simplified by the outage threshold ${{\varepsilon _{out}}}$. First, the actual and maximum achievable SINR are respectively written as
  \begin{equation}
  {\hat {SINR}_{n,i,k}}{\rm{ = }}{{a_{n,i,k}^1} \over {a_{n,i,k}^2}} = {2^{{{{R_{n,i,k}}} \over {{R_{SC}}}}}} - 1,
  \end{equation}

  \begin{equation}
  SIN{R_{n,i,k}}{\rm{ = }}{{b_{n,i,k}^1} \over {b_{n,i,k}^2}},
  \end{equation}
where $b_{n,i,k}^1 = {p_{n,i,k}}|{{\hat H }_{n,i,k}}{|^2}$ and $b_{n,i,k}^2 = p_k^M|H_{n,i,k}^M{|^2} + \sum\limits_{j = 1}^{n - 1} {{p_{j,i,k}}|{{\hat H }_{n,i,k}}{|^2}}  + {\sigma ^2}$. Therefore, the probability that the actual instantaneous rate is more than the maximum achievable capacity is
derived as follows
  \begin{equation}
  \label{CSI}
    \begin{aligned}
      &\Pr [{C_{n,i,k}} < {R_{n,i,k}}|{\hat {g} _{n,i,k}}] \\
      &= \Pr [SIN{R_{n,i,k}}<{\hat {SINR}_{n,i,k}}|{\hat {g} _{n,i,k}}]\\
      &= \Pr [{{b_{n,i,k}^1} \over {b_{n,i,k}^2}} < {2^{{{{R_{n,i,k}}} \over {{R_{SC}}}}}} - 1|{\hat {g} _{n,i,k}}]\\
      &= \Pr [{{b_{n,i,k}^1} \over {b_{n,i,k}^2}} < {2^{{{{R_{n,i,k}}} \over {{R_{SC}}}}}} - 1]\Pr [b_{n,i,k}^1 \le a_{n,i,k}^1|{\hat {g} _{n,i,k}}]\\
      &+ \Pr [{{b_{n,i,k}^1} \over {b_{n,i,k}^2}} < {2^{{{{R_{n,i,k}}} \over {{R_{SC}}}}}} - 1]\Pr [b_{n,i,k}^1 > a_{n,i,k}^1|{\hat {g} _{n,i,k}}].
    \end{aligned}
  \end{equation}

Then there are more strict constraints to satisfy $\Pr [{C_{n,i,k}} < {R_{n,i,k}}|{\hat {g} _{n,i,k}}] \le {\varepsilon _{out}}$ as follows\cite{25Joint,31Energy}
  \begin{equation}
  \label{CSIconstraint1}
  \Pr [b_{n,i,k}^2 \ge a_{n,i,k}^2|{\hat {g}_{n,i,k}}] \le {\varepsilon _{out}}/2,
  \end{equation}
  \begin{equation}
  \label{CSIconstraint2}
  \Pr [b_{n,i,k}^1 \le a_{n,i,k}^1|{\hat {g} _{n,i,k}}] = {\varepsilon _{out}}/2.
  \end{equation}
To prove it, the following is derived
  \begin{equation}
  %\eqalign{
    \begin{split}
      &\Pr [b_{n,i,k}^2 \ge a_{n,i,k}^2|{\hat{ g}_{n,i,k}}]
      =\Pr [{{a_{n,i,k}^1} \over {b_{n,i,k}^2}} \le {{a_{n,i,k}^1} \over {a_{n,i,k}^2}}|{\hat {g} _{n,i,k}}]\\
      &= \Pr [{{a_{n,i,k}^1} \over {b_{n,i,k}^2}} \le {2^{{{{R_{n,i,k}}} \over {{R_{SC}}}}}} - 1|{\hat {g} _{n,i,k}}] \le {\varepsilon _{out}}/2.
    \end{split}
  \end{equation}
Therefore,
  \begin{equation}
  \Pr [{{b_{n,i,k}^1} \over {b_{n,i,k}^2}} < {2^{{{{R_{n,i,k}}} \over {{R_{SC}}}}}} - 1|b_{n,i,k}^1 > a_{n,i,k}^1,{\hat {g} _{n,i,k}}] \le {\varepsilon _{out}}/2.
  \end{equation}
Thanks to the following two formulas
  \begin{equation}
  \Pr [b_{n,i,k}^1 > a_{n,i,k}^1|{\hat {g} _{n,i,k}}] = 1 - {\varepsilon _{out}}/2,
  \end{equation}
  \begin{equation}
  \Pr [{{b_{n,i,k}^1} \over {b_{n,i,k}^2}} < {2^{{{{R_{n,i,k}}} \over {{R_{SC}}}}}} - 1|b_{n,i,k}^1 \le a_{n,i,k}^1,{\hat {g}_{n,i,k}}] \le 1.
  \end{equation}
Therefore, when ${\varepsilon _{out}} \ll 1$ the constraint in (\ref{CSI}) approximates as
  \begin{equation}
  \Pr [{C_{n,i,k}} < {R_{n,i,k}}|{\hat {g} _{n,i,k}}] \le {{{\varepsilon _{out}}} \over 2} + ({{{\varepsilon _{out}}} \over 2})(1 - {{{\varepsilon _{out}}} \over 2}) \approx {\varepsilon _{out}}.
  \end{equation}
 According to the definition of Markov inequality \cite{31Energy,33Interference}, Then (\ref{CSIconstraint1}) and (\ref{CSIconstraint2}) are rewritten as following
  \begin{equation}
  \label{rewriteCSI1}
    \begin{split}
      &\Pr [b_{n,i,k}^2 \ge a_{n,i,k}^2|{\hat {g} _{n,i,k}}]\\
      &= \Pr [\sum\limits_{j = 1}^{n - 1} {{p_{j,i,k}}|{{\hat H }_{n,i,k}}{|^2}}+ {\sigma ^2} + p_k^M|H_{n,i,k}^M{|^2}  \ge a_{n,i,k}^2 |{\hat {g} _{n,i,k}}] \\
      &= \Pr [\sum\limits_{j = 1}^{n - 1} {{p_{j,i,k}}|{{\hat H }_{n,i,k}}{|^2}}  \ge a_{n,i,k}^2 - {\sigma ^2} - p_k^M|H_{n,i,k}^M{|^2}|{\hat {g} _{n,i,k}}] \\
      &\le {{E[\sum\limits_{j = 1}^{n - 1} {{p_{j,i,k}}|{{\hat H }_{n,i,k}}{|^2}} ]} \over {a_{n,i,k}^2 - {\sigma ^2} - p_k^M|H_{n,i,k}^M{|^2}}} = {{\sum\limits_{j = 1}^{n - 1} {{p_{j,i,k}}|{{\hat H }_{n,i,k}}{|^2}} } \over {a_{n,i,k}^2 - {\sigma ^2} - p_k^M|H_{n,i,k}^M{|^2}}},
    \end{split}
  \end{equation}

  \begin{equation}
    \begin{split}
      &\Pr [b_{n,i,k}^1 \le a_{n,i,k}^1|{\hat {g} _{n,i,k}}]\\
      & = \Pr [{p_{n,i,k}}|{{\hat H }_{n,i,k}}{|^2} \le a_{n,i,k}^1|{\hat {g} _{n,i,k}}]\\
      & = \Pr [|{g_{n,i,k}}{|^2} \le {{a_{n,i,k}^1} \over {{p_{n,i,k}}|PL{|^2}}}|{\hat {g} _{n,i,k}}]\\
      & = {F_{|{g_{n,i,k}}{|^2}}}({{a_{n,i,k}^1} \over {{p_{n,i,k}}|PL{|^2}}})\\
      & = 1 - {Q_1}(\sqrt {{{2|{g_{n,i,k}}{|^2}} \over {\sigma _e^2}}} ,\sqrt {{{2a_{n,i,k}^1} \over {\sigma _e^2{p_{n,i,k}}|PL{|^2}}}} )\\
      & = {\varepsilon _{out}}/2,
    \end{split}
  \end{equation}
where $|{g_{n,i,k}}{|^2} \sim CN({\hat {g} _{n,i,k}},{\sigma ^2})$ and ${Q_1}(X,Y)$ is the Marcum Q-function. $PL$ is the path loss related to distance, which is a simple way of writing $PL{(d)_{n,i,k}}$. Therefore
  \begin{equation}
  a_{n,i,k}^1 = F_{|{g_{n,i,k}}{|^2}}^{ - 1}({\varepsilon _{out}}/2) \cdot {p_{n,i,k}}|PL{|^2}.
  \end{equation}
Then (\ref{rewriteCSI1}) is rewritten as
  \begin{equation}
    \begin{split}
    &{{\sum\limits_{j = 1}^{n - 1} {{p_{j,i,k}}|{{\hat H }_{n,i,k}}{|^2}} } \over {a_{n,i,k}^1/({2^{{{{R_{n,i,k}}} \over {{R_{SC}}}}}} - 1) - {\sigma ^2} - p_k^M|H_{n,i,k}^M{|^2}}} \\
    &= {{\sum\limits_{j = 1}^{n - 1} {{p_{j,i,k}}|PL{|^2}(|{{\hat {g} }_{n,i,k}}{|^2} + \sigma _e^2)} } \over {{{F_{|{g_{n,i,k}}{|^2}}^{ - 1}({\varepsilon _{out}}/2) \cdot {p_{n,i,k}}|PL{|^2}} \over {{2^{{{{R_{n,i,k}}} \over {{R_{SC}}}}}} - 1}} - {\sigma ^2} - p_k^M|H_{n,i,k}^M{|^2}}} = {{{\varepsilon _{out}}} \over 2}.
    \end{split}
  \end{equation}
where refering to (\ref{CSIconstraint1}) and this formula equal to ${{{\varepsilon _{out}}} \over 2}$ . Let ${\Theta _{n,i,k}} = {\varepsilon _{out}}({\sigma ^2} + p_k^M|H_{n,i,k}^M{|^2})$, and ${\Psi _{n,i,k}}{\rm{ = }}{|P{L_{n,i,k}}{|^2}}{|{{\hat {g}} _{n,i,k}}{|^2} + \sigma _e^2}$. The SINR considering the outage probability of users with imperfect CSI is given by
  \begin{equation}
  {\hat {SINR} _{n,i,k}} = {{{\varepsilon _{out}}F_{|{g_{n,i,k}}{|^2}}^{ - 1}({\varepsilon _{out}}/2) \cdot {p_{n,i,k}}|PL{|^2}} \over {{\Theta _{n,i,k}} + 2{\Psi _{n,i,k}}\sum\limits_{j = 1}^{n - 1} {{p_{j,i,k}}} }}.
  \end{equation}

Multiple users can occupy one subchannel at the same time because of NOMA. In this paper, ${N_k} $ is limited to no more than two users. Assume channel gain $|{{\hat H }_{2,i,k}}| < |{{\hat H }_{1,i,k}}|$, $\forall i,k$, then energy efficiency problem in UAV cells is expressed as (\ref{CSIobj}), where ${\beta _{j,i,k}},j = 1,2$ is ratio for ${{{p_{j,i,k}}} \over {{p_{i,k}}}}$ and ${{p_{i,k}}}$ is power allocated on $S{C_k}$. Our resource allocation method in the UAV network is discussed by two steps based on channel and power allocation. Then two iterative algorithms are designed to deal with the network energy efficiency optimization for NOMA UAV network.

  \newcounter{TempEqCnt}
\setcounter{TempEqCnt}{\value{equation}}
\setcounter{equation}{38}
\begin{figure*}[hb]
\begin{equation}
\label{CSIobj}
\begin{aligned}
 &\mathop {\max }\limits_{\beta ,P} EE =\sum\limits_{i = 1}^I {\sum\limits_{k = 1}^K {\sum\limits_{n = 1}^{{N_k}} {{{(Pr [{C_{n,i,k}} > {R_{n,i,k}}|{{\hat g }_{n,i,k}}]){R_{n,i,k}}} \over {{p_m} + {p_{n,i,k}}}}} } } \\
      &= \mathop {\max }\limits_{\beta ,P} \sum\limits_{i = 1}^I \sum\limits_{k = 1}^K (1 - {\varepsilon _{out}})({{{B_{sc}}{{\log }_2}({\rm{1 + }}{{F_{|{g_{{\rm{1}},i,k}}{|^2}}^{ - 1}({{{\varepsilon _{out}}} \over 2}) \cdot |P{L_{1,i,k}}{|^2}{\beta _{{\rm{1}},i,k}}{p_{i,k}}} \over {{\sigma ^2} + p_k^M|H_{1,i,k}^M{|^2}}})} \over {{p_m} + {\beta _{{\rm{1}},i,k}}{p_{i,k}}}}
      + {{{B_{sc}}{{\log }_2}({\rm{1 + }}{{{\varepsilon _{out}}F_{|{g_{2,i,k}}{|^2}}^{ - 1}({{{\varepsilon _{out}}} \over 2}) \cdot |P{L_{2,i,k}}{|^2}{\beta _{{\rm{2}},i,k}}{p_{i,k}}} \over {{\Theta _{2,i,k}} + 2{\Psi _{n,i,k}}{\beta _{{\rm{1}},i,k}}{p_{i,k}}}})} \over {{p_m} + {\beta _{{\rm{2}},i,k}}{p_{i,k}}}}),
\end{aligned}
\end{equation}
\end{figure*}

\subsection{Energy-Efficient User Scheduling}

To begin with, the power of each subchannel is set to equal in the UAV cells. NOMA allows each subchannel to serve multiple users at the same time. Only two users are assigned to each subchannel and each user occupies only one subchannel in this paper. Therefore, the user side and the subchannel side are two-side match. Users select the subchannel at first, then the subchannel side selects the user pair. First, the matching choices of users and subchannels are sorted according to the imperfect CSI from UAVs to users. The user prefers $SC{_i}{\rm{ }}$ if the channel gain estimated on the $SC{_i}{\rm{ }}$ is greater than $SC{_j}{\rm{ }}$. As for the subchannel side, if the maximum energy efficiency on this subchannel achieved by user pair ${U_i}$ is greater than ${U_j}$, the user pair ${U_i}$ will be matched to the subchannel. Then our goal is to find optimal user pair to obtain maximum energy efficiency on each subchannel to complete user scheduling.

When considering the match of users and subchannels, because the (\ref{CSIobj}) is non-convex for ${\beta}$,
DC programming which is represented as a minimization of a difference of two convex functions, is used to solve this problem as follows.
  \begin{equation}
  \label{DCaim}
  \mathop {\min }\limits_{\beta  \in (0,1)}  F(\beta ) = \mathop {\min }\limits_{\beta  \in (0,1)} {F_1}(\beta ) - {F_2}(\beta ).
  \end{equation}

Then objective function (\ref{CSIobj}) is rewritten as (\ref{41}) and (\ref{42}), where $\nabla {F_2}({\beta })$ is the gradient of ${F_2}(\beta )$ , ${\Delta _{n,i,k}} = {\varepsilon _{out}}({\sigma ^2} + p_k^M|H_{n,i,k}^M{|^2}){\rm{ + }}{\varepsilon _{out}}F_{|{g_{n,i,k}}{|^2}}^{ - 1}({{{\varepsilon _{out}}} \over 2}) \cdot |P{L_{n,i,k}}{|^2}{p_{i,k}}$.
 % \newcounter{TempEqCnt}
%\setcounter{TempEqCnt}{\value{equation}}
\setcounter{equation}{40}
  \begin{figure*}[h]
  \begin{equation}
  \label{41}
  {F_1}(\beta ) =  - {{{B_{sc}}{{\log }_2}({\rm{1 + }}{{F_{|{g_{{\rm{1}},i,k}}{|^2}}^{ - 1}({{{\varepsilon _{out}}} \over 2}) \cdot |P{L_{1,i,k}}{|^2}{\beta _{{\rm{1}},i,k}}{p_{i,k}}} \over {{\sigma ^2} + p_k^M|H_{1,i,k}^M{|^2}}})} \over {{p_m}+ {\beta _{{\rm{1}},i,k}}{p_{i,k}}}}
  + {{{B_{sc}}{{\log }_2}({\Theta _{2,i,k}} + 2{\Psi _{2,i,k}}{\beta _{{\rm{1}},i,k}}{p_{i,k}})} \over {{p_m} + (1 - {\beta _{1,i,k}}){p_{i,k}}}},
  %\end{split}
  \end{equation}
\end{figure*}

\setcounter{equation}{41}
  \begin{figure*}[h]
  \begin{equation}
  \label{42}
  {F_2}(\beta ) = {{{B_{sc}}{{\log }_2}((2{\Psi _{2,i,k}}) - {\varepsilon _{out}}F_{|{g_{2,i,k}}{|^2}}^{ - 1}({{{\varepsilon _{out}}} \over 2}) \cdot |P{L_{2,i,k}}{|^2}){\beta _{{\rm{1}},i,k}}{p_{i,k}} + {\Delta _{2,i,k}})} \over {{p_m} + (1 - {\beta _{1,i,k}}){p_{i,k}}}},
  \end{equation}
\end{figure*}

\setcounter{equation}{42}
  \begin{figure*}[h]
    \begin{equation}
  %\begin{split}
    \nabla {F_2}({\beta })
    = {B_{sc}}{{{{(2{\Psi _{2,i,k}}) - {\varepsilon _{out}}F_{|{g_{2,i,k}}{|^2}}^{ - 1}({{{\varepsilon _{out}}} \over 2}) \cdot |P{L_{2,i,k}}{|^2}){p_{i,k}}({p_m} + (1 - {\beta _{1,i,k}}){p_{i,k}})} \over {((2{\Psi _{2,i,k}}) - {\varepsilon _{out}}F_{|{g_{2,i,k}}{|^2}}^{ - 1}({{{\varepsilon _{out}}} \over 2}) \cdot |P{L_{2,i,k}}{|^2}){\beta _{{\rm{1}},i,k}}{p_{i,k}} + {\Delta _{2,i,k}})\ln 2}}  } \over {{{({p_m} + (1 - {\beta _{1,i,k}}){p_{i,k}})}^2}}}
   -{B_{sc}}{{{F_2}(\beta )( - {p_{i,k}})}\over {{{({p_m} + (1 - {\beta _{1,i,k}}){p_{i,k}})}^2}}},
  %\end{split}
  \end{equation}
\end{figure*}

A suboptimal approach is designed to achieve user scheduling by replacing $ - {F_2}(\beta )$ in (\ref{DCaim}) with $ - {F_2}(\beta ) - \nabla {F_2}({\beta ^t})$. First the convexity of ${F_1}(\beta )$ and ${F_2}(\beta )$ is proved in Proposition 1 as follows \cite{34Energy}
%lemma
\begin{mypro}
If $-s(\beta )={{B_{sc}}{{\log }_2}({\rm{1 + }}{{F_{|{g_{{\rm{1}},i,k}}{|^2}}^{ - 1}({{{\varepsilon _{out}}} \over 2}) \cdot |P{L_{1,i,k}}{|^2}{\beta _{{\rm{1}},i,k}}{p_{i,k}}} \over {{\sigma ^2} + p_k^M|H_{1,i,k}^M{|^2}}})}$ is strictly concave in $\beta_{1,i,k}$, $-S(\beta )={{{B_{sc}}{{\log }_2}({\rm{1 + }}{{F_{|{g_{{\rm{1}},i,k}}{|^2}}^{ - 1}({{{\varepsilon _{out}}} \over 2}) \cdot |P{L_{1,i,k}}{|^2}{\beta _{{\rm{1}},i,k}}{p_{i,k}}} \over {{\sigma ^2} + p_k^M|H_{1,i,k}^M{|^2}}})} \over {{p_m} + {\beta _{{\rm{1}},i,k}}{p_{i,k}}}}$ is strictly quasiconcave.
\end{mypro}
\begin{proof}
The $\alpha$-sublevel set of function $-S(\beta )$ is expressed as
\begin{equation}
{S_a} = \{ 1 > {\beta _{1,i,k}} > 0| - S({\beta _{1,i,k}}) \ge \alpha \}.
\end{equation}
According to the definition of $\alpha$-sublevel, $- S({\beta _{1,i,k}})$ is strictly concave when ${S_a}$ is strictly convex. $ - S({\beta _{1,i,k}}) > \alpha $ when $ \alpha<0 $. When $\alpha>0$, ${S_a} = \{ 1 > {\beta _{1,i,k}} > 0|0 \ge \alpha ({p_m} + {\beta _{{\rm{1}},i,k}}{p_{i,k}}) + s({\beta _{1,i,k}})\} $ is strictly convex due to the convexity of $s({\beta _{1,i,k}})$.
Therefore, $S(\beta )$ is strictly quasi-convex. Then the other part in ${F_1}(\beta )$ and ${F_2}({\beta ^t})$ are also strictly quasi-convex by a similar proof process.
\end{proof}

Therefore, in this user scheduling scheme, first, the subchannel selects a user with the highest
gain according to the estimated imperfect CSI. If only zero user or one user is assigned to the subchannel, a new user who has not assigned a subchannel is added to this subchannel. Otherwise, any two users of the new user and original two users will form a total of three user pairs, but only one user pair will be assigned to the subchannel. Next the DC algorithm is used to search the optimal power proportion factor of each user pair to obtain the maximum energy
efficiency of each subchannel. The most energy efficient user pair of three pairs is eventually assigned to this subchannel.

\subsection{Energy-Efficient Power Allocation}

On the basis of the subchannel-user matching and ${\beta _{n,i,k}}$ of each subchannel, then different
power is allocated on subchannels to realize maximum energy efficiency in the UAV cell. (\ref{CSIobj}) is rewritten as
  \setcounter{equation}{44}
  \begin{figure*}[h]
    \begin{equation}
    \label{poweraim}
    \begin{split}
      &\mathop {\max }\limits_{{p_{i,k}}} (1 - {\varepsilon _{out}})\sum\limits_{i = 1}^I \sum\limits_{k = 1}^K\{  {{{{B_{sc}}{{\log }_2}({\Theta _{2,i,k}} + 2{\Psi _{2,i,k}}{\beta _{{\rm{1}},i,k}}{p_{i,k}} + {\varepsilon _{out}}F_{|{g_{2,i,k}}{|^2}}^{ - 1}({{{\varepsilon _{out}}} \over 2}) \cdot |P{L_{2,i,k}}{|^2}{\beta _{{\rm{2}},i,k}}{p_{i,k}})} \over {{p_m} + {\beta _{{\rm{2}},i,k}}{p_{i,k}}}}} \\
      & +  {{{{B_{sc}}{{\log }_2}({\rm{1 + }}{{F_{|{g_{{\rm{1}},i,k}}{|^2}}^{ - 1}({{{\varepsilon _{out}}} \over 2}) \cdot |P{L_{1,i,k}}{|^2}{\beta _{{\rm{1}},i,k}}{p_{i,k}}} \over {{\sigma ^2} + p_k^M|H_{1,i,k}^M{|^2}}})} \over {{p_m} + {\beta _{{\rm{1}},i,k}}{p_{i,k}}}}}- {B_{sc}}{{{{{\log }_2}({\Theta _{2,i,k}} + 2{\Psi _{2,i,k}}{\beta _{{\rm{1}},i,k}}{p_{i,k}})} \over {{p_m} + {\beta _{{\rm{2}},i,k}}{p_{i,k}}}}}  \},
    \end{split}
   \end{equation}
   \end{figure*}
The (\ref{poweraim}) is non-convex with respect to ${p_{i,k}}$, which can be proved similarly to \textbf{Proposition 1}. The successive convex approximation method\cite{35Joint} is considered to be a promising approach for handling non-convex problems. In case $x > 0$ and $y > 0$, $ - \log (x) \ge  - \log (y) - {1 \over y}(x - y)$ \cite{32Super}. Referring to the previous inequality relationship, so a part of (\ref{poweraim}) is approximated as
  \begin{equation}
    \begin{split}
    &- {\log _2}({\Theta _{2,i,k}} + 2{\Psi _{2,i,k}}{\beta _{{\rm{1}},i,k}}{p_{i,k}}){\rm{  }}\\
    &\ge  - {\log _2}({\Theta _{2,i,k}} + 2{\Psi _{2,i,k}}{\beta _{{\rm{1}},i,k}}{p_{i,k}}[l]){\rm{ }}\\
    &- {\rm{ }}{{2{\Psi _{2,i,k}}{\beta _{{\rm{1}},i,k}}({p_{i,k}}- {p_{i,k}}[l])} \over {({\Theta _{2,i,k}} + 2{\Psi _{2,i,k}}{\beta _{{\rm{1}},i,k}}{p_{i,k}}[l])\ln 2}}{\rm{            }}.
    \end{split}
  \end{equation}

Then the problem about power is transformed into a convex problem, and the energy efficiency function for power is rewritten as (\ref{SCA}), which subject to $\sum\limits_{i = 1}^I {{p_{i,k}}|H_{w,i,k}^M{|^2} \le {I_k}}$ and ${{N_k} \le 2.}$
    \setcounter{equation}{46}
  \begin{figure*}[h]
    \begin{equation}
  \label{SCA}
    \begin{split}
      & \mathop {\min }\limits_{{p_{i,k}}}(1 - {\varepsilon _{out}})\sum\limits_{i = 1}^I \sum\limits_{k = 1}^K
      \{  -  {{{{B_{sc}}{{\log }_2}({\Theta _{2,i,k}} + 2{\Psi _{2,i,k}}{\beta _{{\rm{1}},i,k}}{p_{i,k}} + {\varepsilon _{out}}F_{|{g_{2,i,k}}{|^2}}^{ - 1}({{{\varepsilon _{out}}} \over 2}) \cdot |P{L_{2,i,k}}{|^2}{\beta _{{\rm{2}},i,k}}{p_{i,k}})} \over {{p_m} + {\beta _{{\rm{2}},i,k}}{p_{i,k}}}}}  \\
      &  - {{{{B_{sc}}{{\log }_2}({\rm{1 + }}{{F_{|{g_{{\rm{1}},i,k}}{|^2}}^{ - 1}({{{\varepsilon _{out}}} \over 2}) \cdot |P{L_{1,i,k}}{|^2}{\beta _{{\rm{1}},i,k}}{p_{i,k}}} \over {{\sigma ^2} + p_k^M|H_{1,i,k}^M{|^2}}})} \over {{p_m} + {\beta _{{\rm{1}},i,k}}{p_{i,k}}}}} {\rm{ + }}{B_{sc}} {{{{{\log }_2}({\Theta _{2,i,k}} + 2{\Psi _{2,i,k}}{\beta _{{\rm{1}},i,k}}{p_{i,k}}[l])} \over {{p_m} + {\beta _{{\rm{2}},i,k}}{p_{i,k}}}}}  \\
      &+ {B_{sc}} {{{2{\Psi _{2,i,k}}{\beta _{{\rm{1}},i,k}}} \over {{\Theta _{2,i,k}} + 2{\Psi _{2,i,k}}{\beta _{{\rm{1}},i,k}}{p_{i,k}}[l])\ln 2}} \cdot {{{p_{i,k}} - {p_{i,k}}[l]} \over {{p_m} + {\beta _{{\rm{2}},i,k}}{p_{i,k}}}} } \}.
    \end{split}
  \end{equation}
   \end{figure*}

%Algorithm Design
\section{Resource Optimization Algorithm Design}

According to the above analysis, two suboptimal algorithms consisting of user scheduling and cross-subchannel power allocation for performance optimization of the UAV network are proposed as follows.

\subsection{User Scheduling Algorithm}

\setcounter{algorithm}{0}
{\renewcommand\baselinestretch{1.0}\small
\begin{algorithm}[!h]\vspace{-1pt}
\caption{User Scheduling Algorithm}
\begin{algorithmic}[1]{\small}
\STATE  \textbf{Initialization:} The matching priority order of users and subchannels based on the imperfect CSI from the UAV to users.
%\REPEAT
%\STATE  \textbf{User Schedule Algorithm}
\STATE  Users first select the subchannel according to priority order, then the user set
${N_k}$ and $\mathop {{N_k}}\limits^ \sim  $ are formed for occupying or not occupying subchannel $SC{_k}$.
\WHILE {$\mathop {{N_k}}\limits^ \sim   \ne {\rm{0 }}$}
\FOR    {each UAV user}
\IF     {${N_k} < 2$}
\STATE  Assign a new user to the $SC{_k}$ and remove it from $\mathop {{N_k}}\limits^ \sim  $.
\ENDIF
\IF     {${N_k} = 2$}
\STATE  {Find the user pair according to the DC algorithm as \textbf{Algorithm 2} to maximize energy efficiency on the subchannel. The new user pair is assigned to the subchannel, and the other user is placed in the $\mathop {{N_k}}\limits^ \sim  $}.
\ENDIF
\ENDFOR
\ENDWHILE
\end{algorithmic}
\end{algorithm}
\par}

\setcounter{algorithm}{1}
{\renewcommand\baselinestretch{1.0}\small
\begin{algorithm}[!h]\vspace{-1pt}
\caption{DC Algorithm}
\begin{algorithmic}[1]{\small}
\STATE  \textbf{Initialization:} Two users are allocated same power in a subchannel and the iteration number $t$ begins from zero, that is $\beta _{1,l,k}^0=0.5$.
%\REPEAT
%\STATE  \textbf{DC Algorithm}
\FOR    {each UAV cell }
\FOR    {each subchannel }
\WHILE  {$ |F({\beta ^{(t + 1)}}) - F({\beta ^{(t )}})| > \delta $}
\STATE  ${F^{(t)}}(\beta ) = {F_1}(\beta ) - {F_2}({\beta ^{(t)}}{\rm{)}} - \nabla {{F_2}^T}({\beta ^{(t)}}{\rm{) (}}\beta  - {\beta ^{(t)}}{\rm{)}}$
\STATE  Search ${\beta }$ to minimize the objective function ${F^{(t)}}(\beta )$
\STATE  ${\beta ^{(t)}} \leftarrow {\beta }$
\STATE  $t \leftarrow t+1$
\ENDWHILE
\ENDFOR
\ENDFOR
\end{algorithmic}
\end{algorithm}
\par}

In Algorithm 1, users and subchannels first perform a two-side matching according to the
gain of imperfectly CSI. Then the DC algorithm in Algorithm 2 is used to select the user pair to maximizes the energy efficiency by finding the optimal power proportion factor for each user
pair in the same subchannel.

\subsection{Power Allocation Algorithm}

\setcounter{algorithm}{2}
{\renewcommand\baselinestretch{1.0}\small
\begin{algorithm}[!h]\vspace{-1pt}
\caption{Power Allocation Algorithm}
\begin{algorithmic}[1]{\small}
\STATE  \textbf{Initialization:}  Equal powers are assigned to different subchannels. Number of iterations $t$ starts from zero to maximum $T$.
\REPEAT
\STATE  \textbf{Successive Convex Approximation Algorithm}
\STATE  Obtain $p_{i,k}[{(t+1)}]$ by minimizing the objective function ${Z^t}(P)$ which is the $t$th iteration of (\ref{SCA}).
\STATE  Check the constraints of interference and power.
\STATE  $t=t+1$.
\UNTIL {{$|Z(p_{i,k}^{(t + 1)}) - Z(p_{i,k}^{(t)})| > \varsigma $} or $t = T$}
\end{algorithmic}
\end{algorithm}
\par}

Since the power of different subchannels is equal in the above user scheduling, Algorithm 3 is
proposed as an energy efficient cross-subchannel power allocation algorithm in the UAV network.
The problem is converted into a convex function using the successive convex approximation
method based on the existing subchannel-user matching state. Then an iterative algorithm is
designed to find suboptimal solutions under UAV power constraint and interference constraint
for macro users.

A sequential quadratic programming is used to address convex optimization problem in this paper. First, ${p_{i,k}}[0]$ is given an initial value to solve the problem (\ref{SCA}).
Because macro users are affected by cross-layer interference from UAV. At that time macro BS updates the interference constraint $\sum\limits_{i = 1}^I {{p_{i,k}}|H_{w,i,k}^M{|^2} \le {I_k}}$,  ${{p_{i,k}}|H_{w,i,k}^M{|^2}}$ is obtained from the feedback of macro users. Then, the macro BS send the cross-tier interference to UAV via backhaul. Then the solution is treated as the next round initial value ${p_{i,k}}[t]$ with $t=t+1$. Finally, the algorithm terminates when the function converges.

\subsection{Complexity Analysis}

For user scheduling, exhaustive search method can obtain optimal match for all users and subchannels. Suppose there are $I$ UAV cells, $N$ users in each UAV cell and $K=N/2$. Then the complexity is given by $O({{N!} \over {{2^K}}})$ for exhaustive search. The worst situation of Algorithm 1 using DC algorithm is that any new user need to compare with previous two users for all subchannels. Therefore, the complexity is no more than $O({K^2})$ for Algorithm 1.  Logarithmic form is used to represent computational complexity, $O(\ln K) < O(\ln ((2K)!)) = O(\ln ((2K)!) - K)$. Therefore, our user scheduling algorithm reduced computational complexity.

\section{Numerical Simulation And Analysis}

The results of simulation are given to demonstrate effective user scheduling and power allocation
in NOMA UAV network with imperfect CSI. The whole network consists of one ground BS,
UAVs and users. Some users are still provided communication service by the ground macro BS,
while others are provided service by hovering UAVs nearby. The BS is arranged at the center
point with 1000 m for macrocell radius, and 20 macro users are randomly distributed around the
BS receiving signals from it. The paper simulates the scenario where several UAVs randomly
move over the users in need and hover to provide communication services. 4 UAVs are randomly
distributed with radius of 350 m and height of 200 m, and $N$ users need communication
connection service in each UAV cell. The flying height of low-altitude rotor UAVs is about
hundreds of meters. Each UAV has a minimum distance constraint from the ground BS to ensure that the UAV serves users at the edge of the macrocell. And enough users are randomly deployed within the coverage of each UAV, and the number of UAV users in the simulation ranges from 10 to 40.
Carrier frequency is 2 GHz and divided into $K$ subchannels where $K=N/2$.
The minimum distance between the BS and UAV cells is 50 m on a horizontal plane. Different
from the channel fading model from ground BS to users, the channel path loss model in UAV
network includes LOS and NLOS. The AWGN power spectral density ${N_0}$ = -174 dBm/Hz and
AWGN power ${\sigma ^2} = {B \over K}{N_0}$. The maximum power is 5 W for each UAV and 0.5 W is used as
hovering power consumption ${p_m}$. The error tolerance $\varsigma  = \delta  = 0.01$.

Fig. 2 denotes energy efficiency of the UAV network with $N$ from 10 to 40 per UAV cell
in NOMA and OFDMA condition. Outage probability is 0.05. Fig. 2 proves that network energy efficiency rises with increased number of users who need additional UAV services. Fractional transmit power allocation (FTPA) is a common algorithm for user scheduling where power allocated on users in the same subchannel is according to the exponent of subchannel gain.
The new power allocation scheme with successive convex approximation algorithm has a better
performance compared with the NOMA-DC and NOMA-FTPA scheme in \cite{36Energy}. The NOMA-proposed
in the simulation is our resource allocation method. NOMA-DC is the DC algorithm
used for power allocation, which is different from the NOMA-proposed. The NOMA-FTPA
is the resource allocation algorithm with FTPA and NOMA for comparison. And OFDMA method only allows each subchannel
to serve one user at the same time. It is obvious that the energy efficiency of UAV networks
using NOMA is much higher than that of OFDMA. Energy efficiency of our algorithm has
a 3.6$\%$ improvement than the power allocation algorithm in \cite{36Energy} and 17$\%$ more than FTPA
algorithm with 30 users per UAV cell.

\begin{figure}
        \centering
        \includegraphics*[width=8cm]{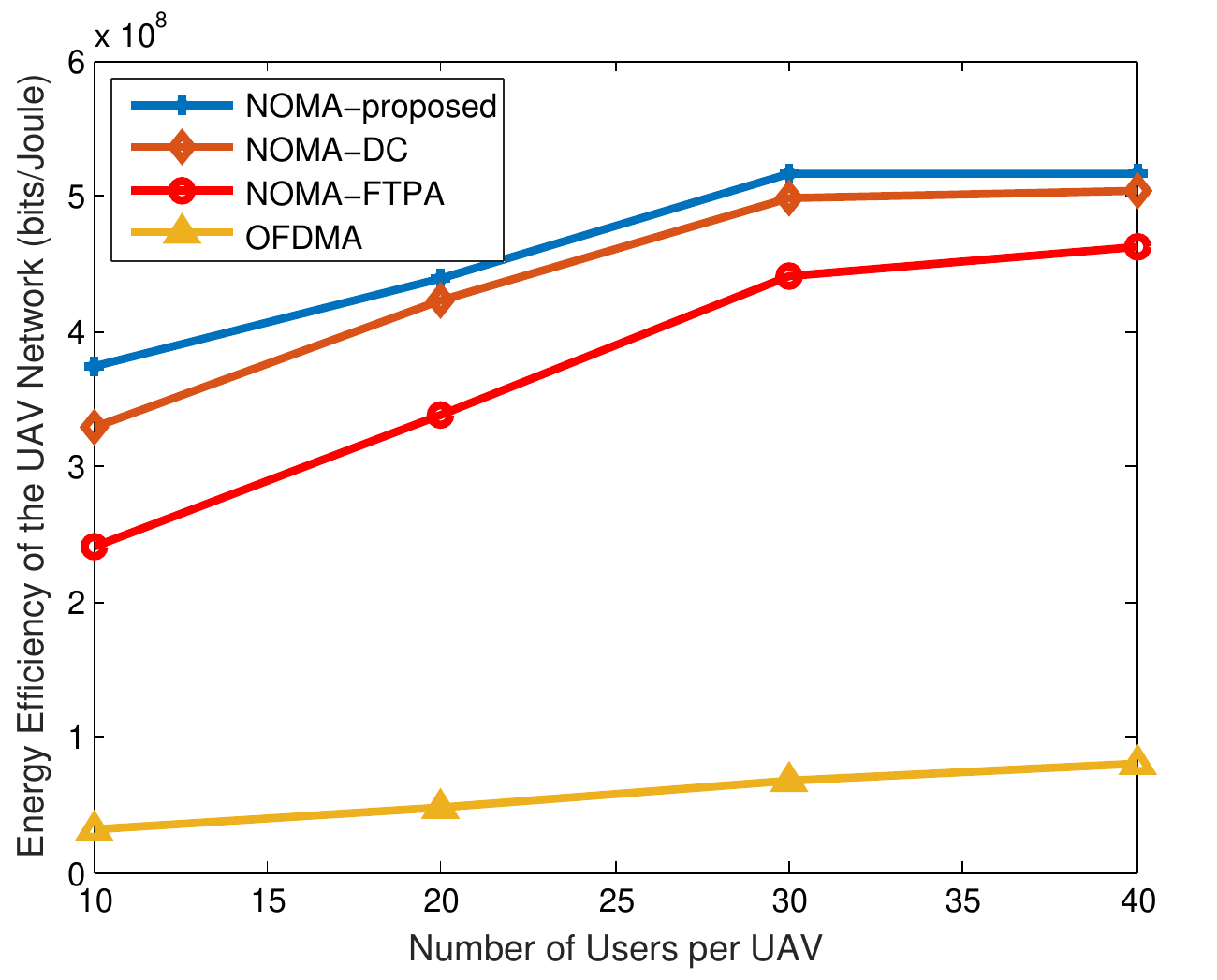}
        \caption{Energy efficiency vs the number of users.}
        \label{fig:2}

\end{figure}

Fig. 3 shows that the energy efficiency varies along with iterations for power allocation
algorithm with 10 and 30 users per UAV cell. Fig. 3 validates the power allocation scheme using successive convex approximation in this paper, we can see that the energy efficiency reaches a stable value after 4 iterations. Therefore, this method can finally obtain a stable suboptimal solution. Besides, system has a higher energy efficiency under 30 users compared to 10 users per UAV cell.

\begin{figure}
        \centering
        \includegraphics*[width=8cm]{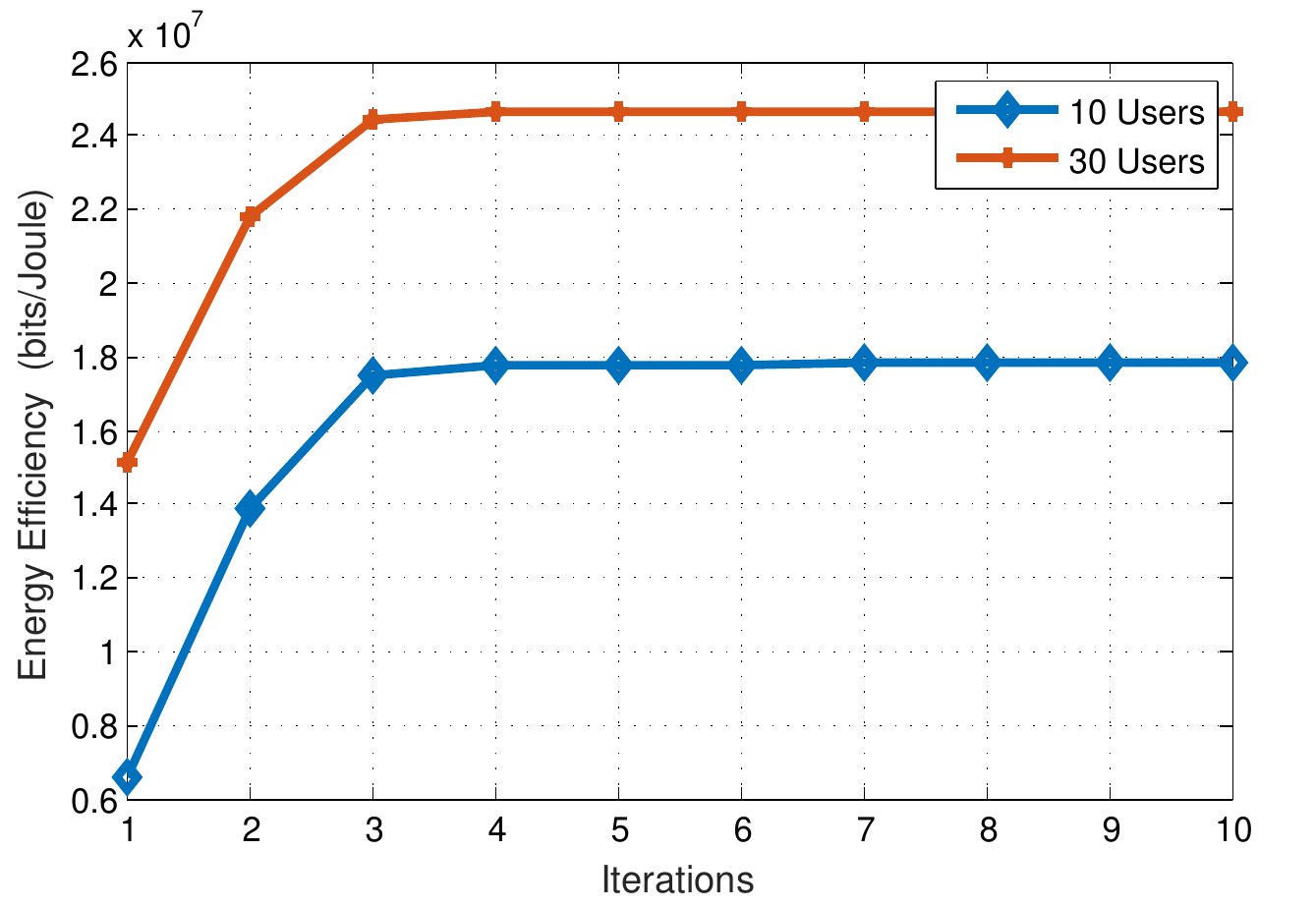}
        \caption{Energy efficiency under iterations.}
        \label{fig:3}

\end{figure}

In Fig. 4, the UAV network performance is evaluated against the number of users and different
estimation errors. NOMA-proposed is the successive convex approximation method used in power
allocation in the paper. And the network energy efficiency declines as the estimation error
increases. Because when the estimation error increases, the ability of the receiving terminal
to sense the UAV¡¯s channel information is poor, the maximum capacity of the channel cannot
be accurately calculated, and the instantaneous rate may exceed the maximum capacity, resulting
in an increased outage probability. In particular, when there are 30 users per UAV cell, the total
UAV network performance with $\sigma _e^2 = 0.01$ improves by 2.9$\%$ over $\sigma _e^2 = 0.05$ as well as by 6$\%$ over $\sigma _e^2 = 0.05$.

\begin{figure}
        \centering
        \includegraphics*[width=8cm]{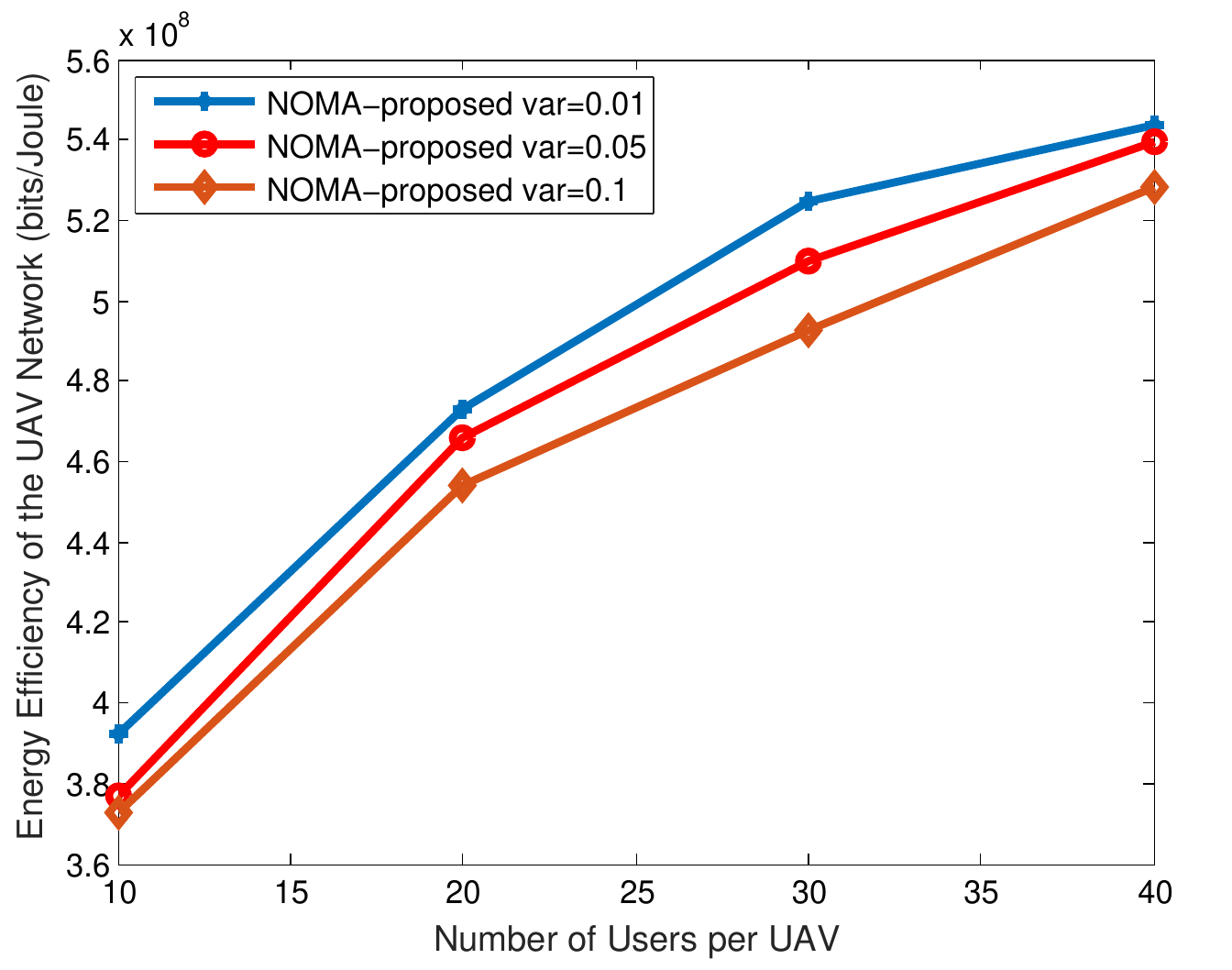}
        \caption{Energy efficiency vs different number of users and estimation errors.}
        \label{fig:4}

\end{figure}

Fig. 5 evaluates energy efficiency of UAV cells against the number of users considering
different CSI. UAVs move to different hovering positions in the air and transmit through LOS
and NLOS communication channels, and eventually lead to different channel errors. The perfect
CSI indicates that there is no channel estimation error, and the user can perceive accurate channel
gain, that is, ${e_{n,i,k}} = 0$ . The energy efficiency in our proposed algorithm is better than FTPA
algorithm regardless perfect CSI or imperfect CSI with $\sigma _e^2 = 0.2$. When there are 40 users per
UAV cell, the total energy efficient with perfect CSI is 6$\%$ more than $\sigma _e^2 = 0.2$ in our proposed
algorithm.

\begin{figure}
        \centering
        \includegraphics*[width=8cm]{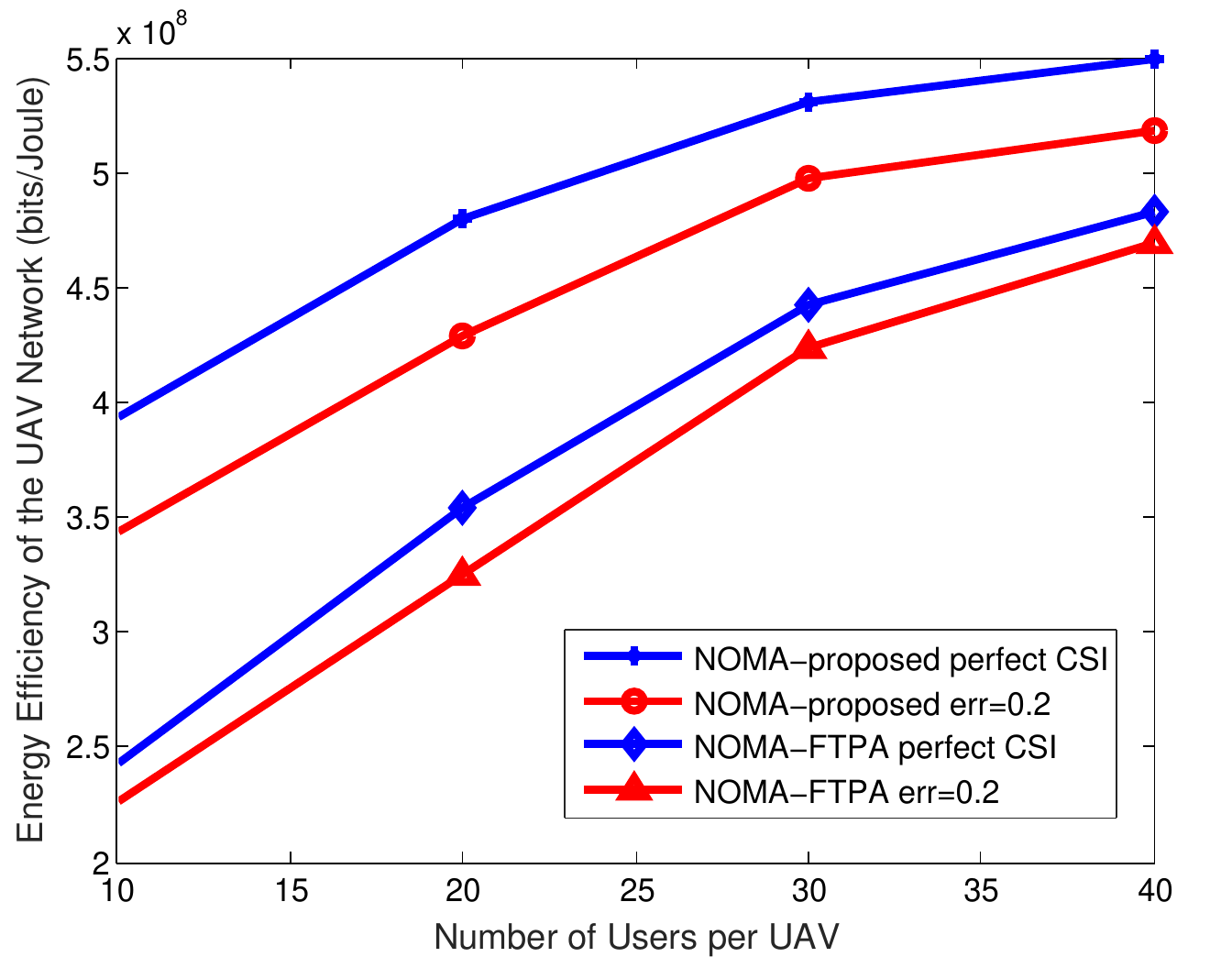}
        \caption{Energy efficiency under perfect and imperfect CSI.}
        \label{fig:5}

\end{figure}

Fig. 6 evaluates the sum of UAV network energy efficiency against hovering power consumption
${p_m}$. 10 users are set in each UAV cell and the maximum transmit power of each UAV is 10 W.
It can be obtained that the UAV network performance is degraded as the ${p_m}$ increases in NOMA
and OFDMA. The increase in hovering power causes an increase in overall power consumption,
and the ratio of the total data rate to it decreases. However, the performance of our proposed
algorithm still outperforms the the existing schemes. The energy efficiency in our algorithm exceed 19$\%$ than FTPA and is far superior to OFDMA when the power consumption of ${p_m}$
is 1 W.

\begin{figure}
        \centering
        \includegraphics*[width=8cm]{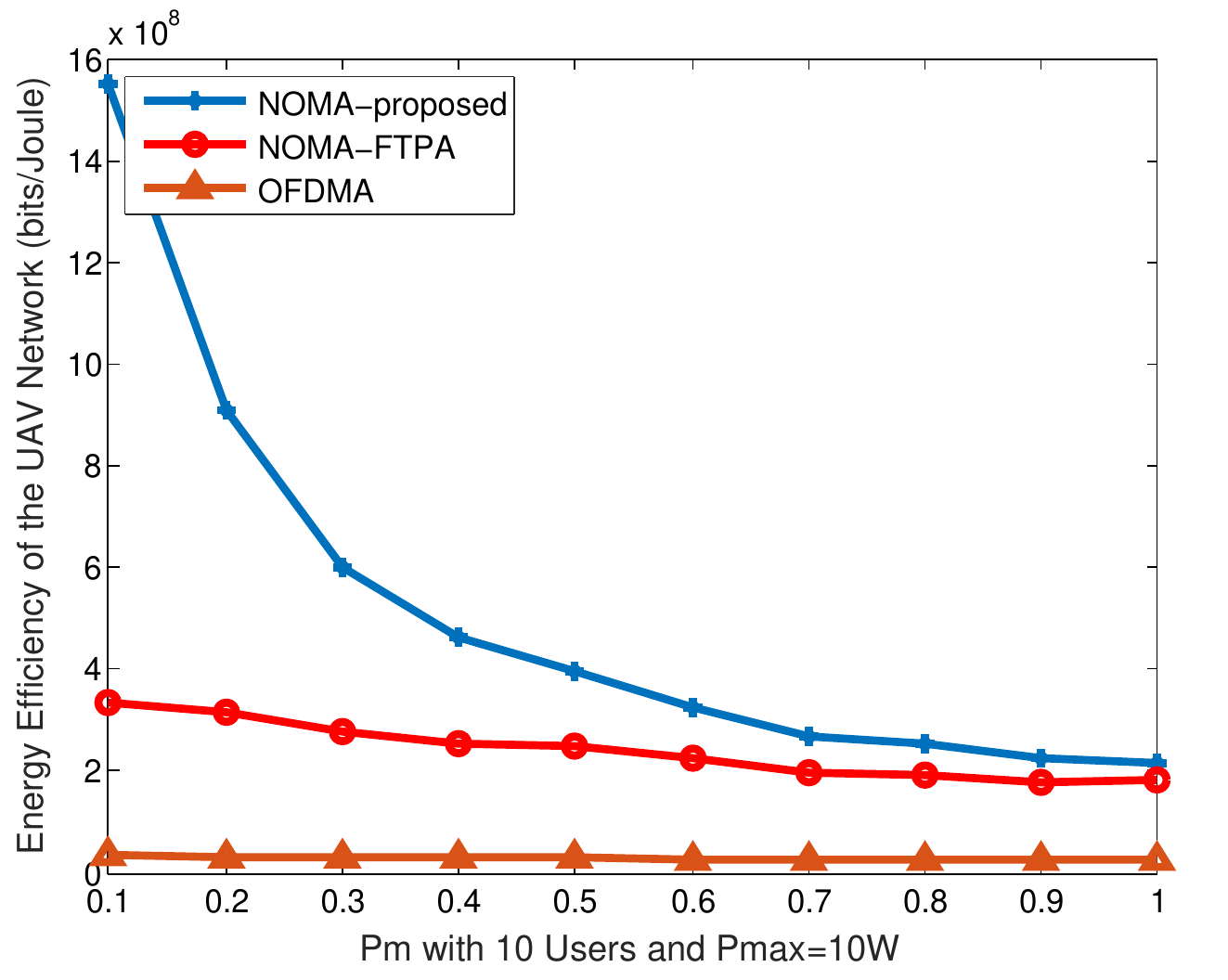}
        \caption{Energy efficiency vs ${p_m}$ with ${P_{\max }}$ = 10 W.}
        \label{fig:6}

\end{figure}

Fig. 7 evaluates the whole energy efficiency of UAV network versus hovering power consumption
${p_m}$ with 5 W for the maximum transmit power per UAV. There are 10 users per UAV cell
and estimation error is 0.05. The energy efficiency degrades as the ${p_m}$ increases. For example,
the performance of our proposed algorithm improves 18$\%$ more than FTPA algorithm in user
scheduling with hovering power consumption 0.5 W.

\begin{figure}
        \centering
        \includegraphics*[width=8cm]{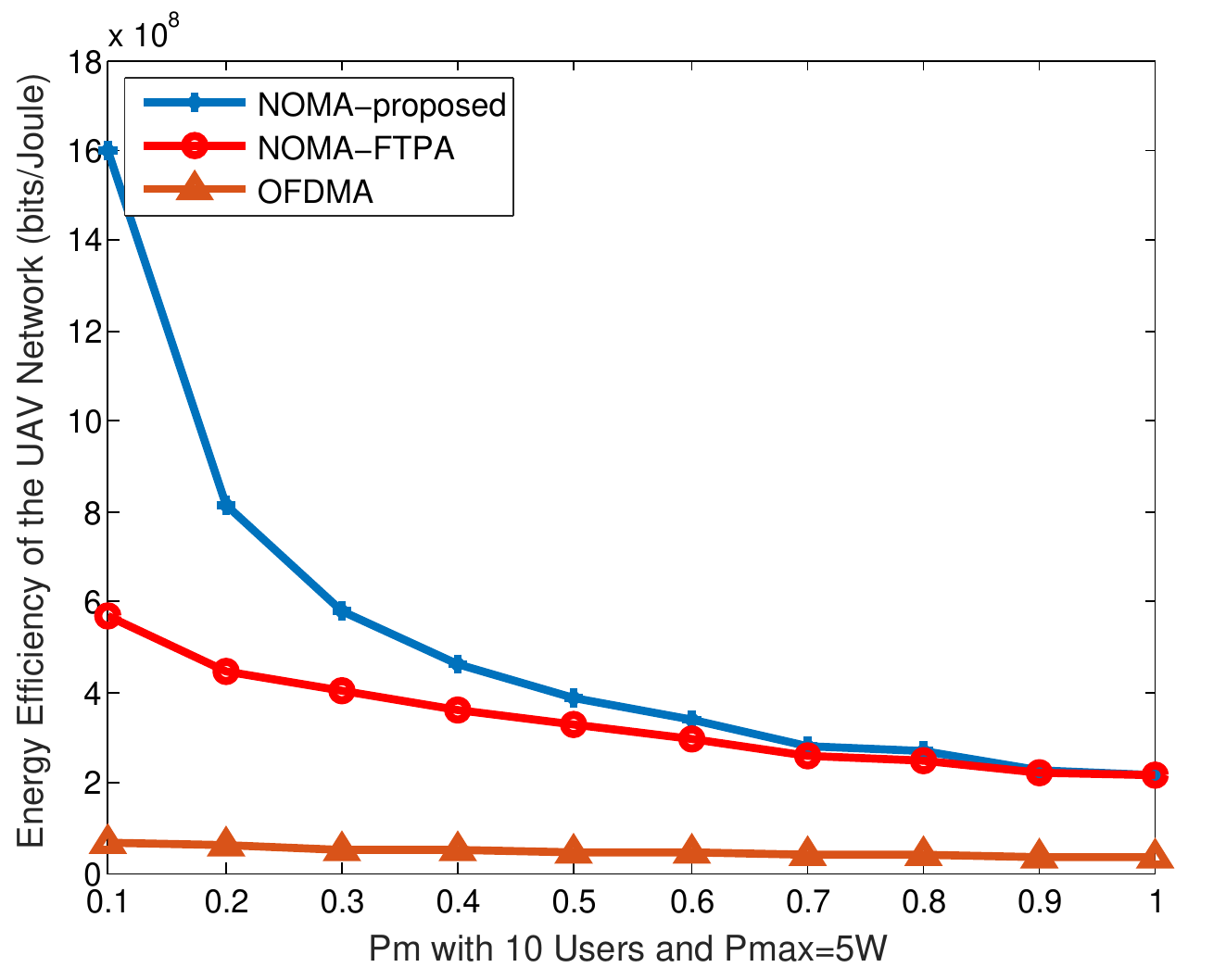}
        \caption{Energy efficiency of the system vs ${p_m}$ with ${P_{\max }}$ = 5 W.}
        \label{fig:7}

\end{figure}

The energy efficiency of UAV network versus hovering power consumption ${p_m}$ with different
estimation errors of low-altitude channel is revealed in Fig. 8. System performance decreases
not only as hovering power consumption increases but also as estimation error increases. At the
same hovering power consumption, the subchannel with a smaller estimation error can achieve
greater energy efficiency. Because the estimation error means the channel perception capability.
For ${p_m}$ = 0.5 W with 10 users per UAV cell, the system performance with $\sigma _e^2 = 0.01$ improves
5$\%$ than $\sigma _e^2 = 0.1$ and 19.7$\%$ than $\sigma _e^2 = 0.5$.

\begin{figure}
        \centering
        \includegraphics*[width=8cm]{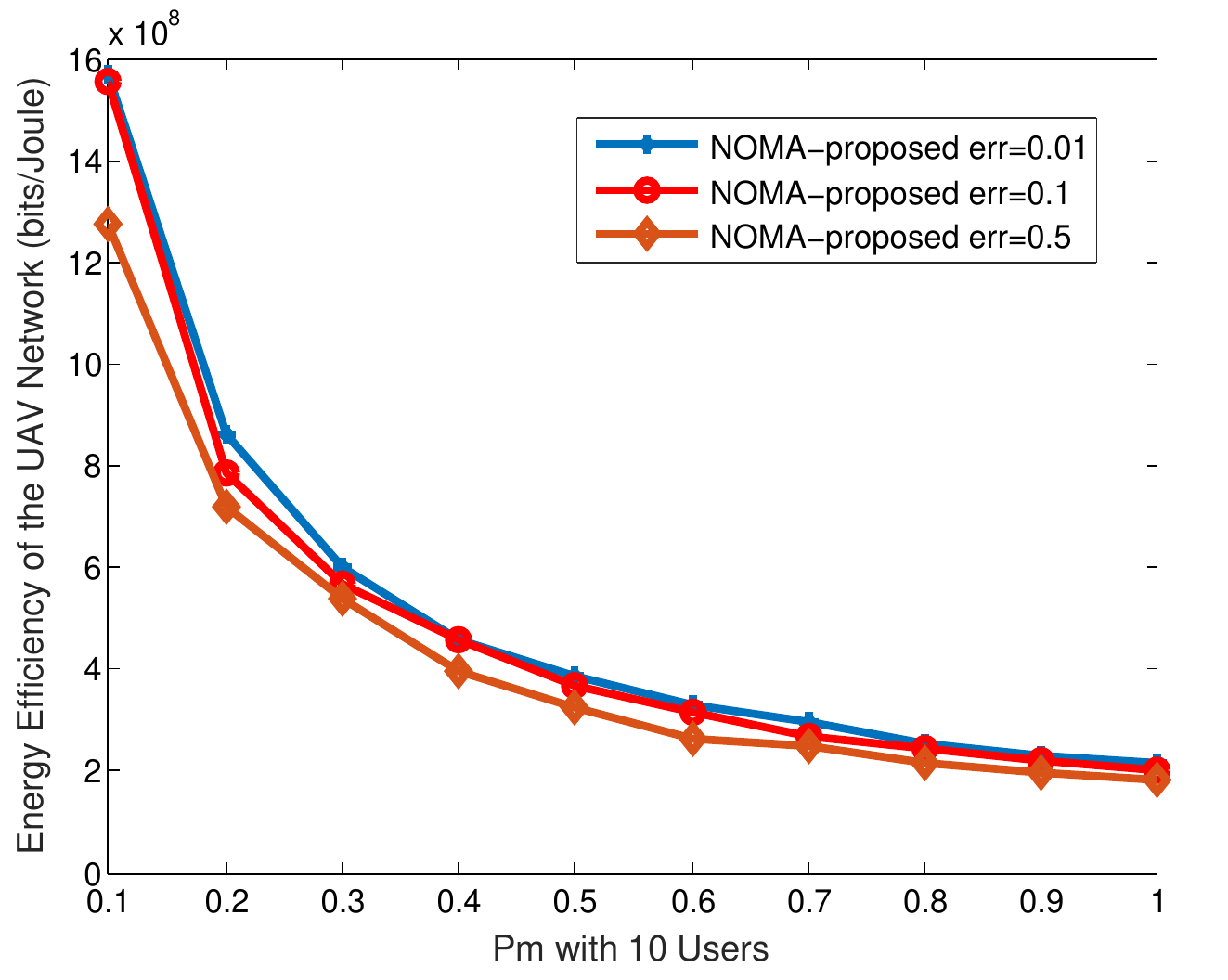}
        \caption{Energy efficiency vs ${p_m}$ under different estimation errors.}
        \label{fig:8}

\end{figure}

The network performance versus the height of UAVs is given in Fig. 9. The channel estimation
error is 0.05 and 20 users per UAV cell. The mobility of UAV in the vertical direction will also
bring changes in network energy efficiency. Regardless of the algorithm, the system energy
efficiency increases and then drops as the height of UAV grows. The distance from UAVs
to users becomes farther with a higher UAV, but the probability of LOS becomes larger too.
Therefore, path loss is reduced due to the larger LOS probability at first. However as the height
continues to increase, the distance influence become greater and the overall path loss increases.
The height value for realizing maximum energy efficiency is related to the environment in which
the UAV and the user are located. When the height of UAV is 150 m, the energy efficiency of
our algorithm raises 8.8$\%$ in comparison with NOMA-FTPA algorithm and increases 34$\%$ than
OFDMA algorithm.

\begin{figure}
        \centering
        \includegraphics*[width=8cm]{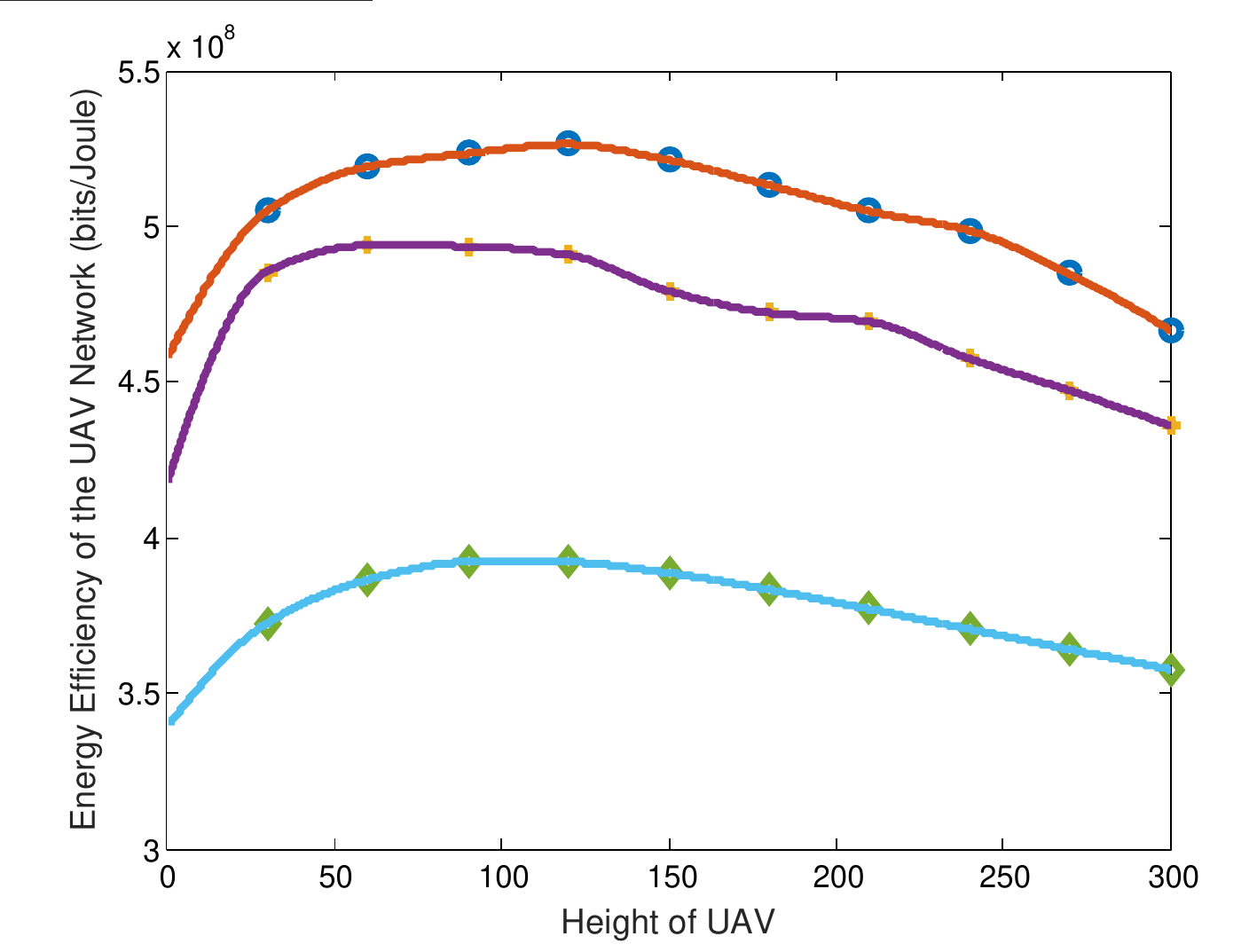}
        \caption{Energy efficiency vs height of UAVs.}
        \label{fig:9}

\end{figure}

\section{Conclusions}

Resource allocation to achieve maximum energy efficiency in the NOMA UAV network was
studied in our paper. Because UAVs can move and provide services in the air, the constraints
of the energy efficiency optimization problem included UAV power, outage probability, and
interference limitation to macro users. Due to the small-scale fading in low-altitude channels,
we considered imperfect CSI for UAV communications. The non-convex target problem with
imperfect CSI was transformed into a problem without probability constraint at first. Resource
allocation was divided into two steps with a fixed hovering height of the UAV in the UAV
network. In user scheduling, a suboptimal algorithm was proposed to achieve subchannels and
users matching by finding power proportion factors for users assigned on the same subchannel.
Then the objective function was non-convex for power yet and was converted into a convex
problem through the successive convex approximation method. Finally, a suboptimal algorithm
was designed to handle the convex problem. The simulation proved that energy efficiency in the UAV network was promoted by our algorithm in comparison with existing algorithms.

\end{document}